\documentclass[orivec]{llncs}

\usepackage[T1]{fontenc}
\usepackage[utf8x]{inputenc}
\usepackage[english]{babel}

\usepackage{mathtools,amssymb}

\usepackage{enumitem}
\setenumerate{label={{\rm({\roman*})}},leftmargin=6mm,itemsep=3pt,topsep=3pt}
\setitemize{label={$\vcenter{\hbox{\tiny$\bullet$}}$},leftmargin=6mm,itemsep=3pt,topsep=3pt}

\usepackage{multirow}
\usepackage[caption=false]{subfig}
\usepackage[outline]{contour}
\contourlength{1.15pt}
\usepackage{ifthen}
\usepackage{tikz}
\usetikzlibrary{calc}
\tikzstyle{every picture}=[line join=round,line width=.75pt,every label/.append style={font=\scriptsize},label distance=-1.5pt]
\tikzstyle{every node}=[font=\scriptsize]
\usepackage{lmodern}

\newcommand{\R}{\mathbb R}

\renewcommand{\dim}{\mathsf{dim}}

\newcommand{\scalProd}[2]{\langle{#1},{#2}\rangle}
\newcommand{\N}{N} % for the neighborhood operator

\newcommand{\conv}{\mathsf{conv}}

\newcommand{\fool}{\mathsf{fool}}
\newcommand{\supp}{\mathsf{supp}}
\newcommand{\rc}{\mathsf{rc}}

\newcommand{\pp}{\mathsf{p}}
\newcommand{\LL}{\mathsf{\ell}}
\newcommand{\cc}{\mathsf{c}}
\newcommand{\xx}{\mathsf{x}}

\newcommand{\lines}{\mathcal{L}}
\newcommand{\points}{\mathcal{P}}

\DeclareMathOperator{\edge}{\mathsf P_\mathsf{edge}}
\DeclareMathOperator{\stab}{\mathsf{STAB}}
\DeclareMathOperator{\xc}{\mathsf{xc}}
\DeclareMathOperator{\stp}{P_{\mathsf{sp.trees}}}
\newcommand{\PG}{\mathsf{PG}}
\newcommand{\GF}{\mathsf{GF}}

\spnewtheorem{lem}[theorem]{Lemma}{\bfseries}{\itshape}
\spnewtheorem{defn}[theorem]{Definition}{\bfseries}{\rmfamily}
\spnewtheorem{clm}[theorem]{Claim}{\bfseries}{\itshape}
\spnewtheorem*{subproof}{Subproof}{\itshape}{\rmfamily}

\makeatletter
\renewcommand{\@Opargbegintheorem}[4]{%
  #4\trivlist\item[\hskip\labelsep{#3#2\@thmcounterend}]}
\makeatother

\begin{document}

\title{Extension complexity of stable set polytopes of bipartite graphs}
\author{Manuel Aprile,\inst{1} Yuri Faenza,\inst{2} Samuel Fiorini,\inst{3} Tony Huynh,\inst{3} Marco Macchia\inst{3}}

\authorrunning{M.\ Aprile et al.} % abbreviated author list (for running head)

\institute{\'Ecole Polytechnique F\'ed\'erale de Lausanne (EPFL), Lausanne, Switzerland
\email{manuel.aprile@epfl.ch}
\and
IEOR Department, Columbia University, New York, USA
\email{yf2414@columbia.edu}
\and
Universit\'e libre de Bruxelles, Brussels, Belgium
\email{\{sfiorini,mmacchia\}@ulb.ac.be}
\email{tony.bourbaki@gmail.com}}
\date{}
\maketitle

\begin{abstract}
 The \emph{extension complexity} $\xc(P)$ of a polytope $P$ is the minimum number of facets of a polytope that affinely projects to $P$. Let $G$ be a bipartite graph with $n$ vertices, $m$ edges, and no isolated vertices. Let $\stab(G)$ be the convex hull of the stable sets of $G$. It is easy to see that $n \leqslant \xc (\stab(G)) \leqslant n+m$. We improve both of these bounds. For the upper bound, we show that $\xc (\stab(G))$ is $O(\frac{n^2}{\log n})$, which is an improvement when $G$ has quadratically many edges. For the lower bound, we prove
 that $\xc (\stab(G))$ is $\Omega(n \log n)$ when $G$ is the incidence graph of a finite projective plane. We also provide examples of $3$-regular bipartite graphs $G$ such that the edge vs stable set matrix of $G$ has a fooling set of size $|E(G)|$.
\end{abstract}

\section{Introduction}
A polytope $Q \subseteq \R^p$ is an \emph{extension} of a polytope $P \subseteq \R^d$ if there exists an affine map $\pi: \R^p \rightarrow \R^d$ with $\pi(Q) = P$. The \emph{extension complexity} $\xc(P)$ of $P$ is the minimum number of facets of any extension of $P$. If $Q$ is an extension of $P$ such that $Q$ has significantly fewer facets than $P$, then it is advantageous to run linear programming algorithms over $Q$ instead of $P$.

%The \emph{extension complexity} $\xc(P)$ of a polytope $P$ is the minimum number of facets of a polytope $Q$ such that there exists an affine map $\pi$ with $\pi(Q) = P$. This is an important complexity measure in combinatorial optimization: if $Q$ has significantly fewer facets than $P$, we can run linear programming algorithms over $Q$ instead of $P$.

One example of a polytope that admits a much more compact representation in a higher dimensional space is the spanning tree polytope, $\stp (G)$. Edmonds' \cite{Edmonds71} classic description of $\stp (G)$ has $2^{\Omega(|V|)}$ facets. However, 
Wong \cite{Wong80} and Martin \cite{Martin91} proved that for every connected graph $G = (V,E)$,  
\[
|E| \leqslant \xc(\stp(G)) \leqslant O(|V| \cdot |E|).
\]

Fiorini, Massar, Pokutta, Tiwary, and de Wolf \cite{FMPTW15} were the first to show that many polytopes arising from $\mathsf{NP}$-hard problems (such as the stable set polytope) do indeed have high extension complexity.  Their results answer an old question of Yannakakis \cite{Yannakakis91} and do not rely on any complexity assumptions such as $\mathsf P \neq \mathsf{NP}$.

On the other hand, Rothvo\ss~\cite{Rothvoss14} proved that the perfect matching polytope of the complete graph $K_n$ has extension complexity at least $2^{\Omega (n)}$.  This is somewhat surprising since the maximum weight matching problem can be solved in polynomial-time via Edmond's blossom algorithm \cite{Edmonds65}.  By now many accessible introductions to extension complexity are available (see \cite{Kaibel11}, \cite{CGZ13}, \cite{CCZ14}, \cite{Roughgarden15}).

% \begin{figure}[ht!]\centering
% \begin{tikzpicture}[scale=1]
% %\clip(-1.5,-1.8) rectangle (8,2.5);
% \tikzset{ext/.style={line width=1.35pt,fill=lightblue!20!,draw=lightblue,fill opacity=0.85},
% poly/.style={line width=1.35pt,fill=Red2!20!,draw=Red2},
% lin/.style={line width=1pt,dotted,draw=Red2},
% proj/.style={line width=1pt,dashed,draw=gray,opacity=0.65}}

% \begin{scope}[line width=1.35pt,->]
% \draw(0.25,0,0.25)--(0.25,0,4);
% \draw(0.25,0,0.25)--(4,0,0.25);
% \end{scope}

% \draw[line width=1.7pt,->,darkgray](4.5,2.75,2.5)--(4.5,1,2.5) node[midway,right] {$\pi$};

% \draw[ext](3,1.25,1)--(2,1.25,1)--(1,3,2)--(4,3,2)-- cycle;
% \draw[ext](1,3,2)--(1,3,3)--(2,1.25,4)--(2,1.25,1)-- cycle;
% \draw[ext](2,1.25,1)--(2,1.25,4)--(3,1.25,4)--(3,1.25,1)-- cycle;

% \draw[poly](2,0,1)--(3,0,1)--(4,0,2)--(4,0,3)--(3,0,4)--(2,0,4)--(1,0,3)--(1,0,2)-- cycle;

% \draw (2.6,0,2.5) node{\large\color{Red2}$P$};

% \draw[proj](2,0,1)--(2,1.25,1);
% \draw[proj](3,0,1)--(3,1.25,1);
% \draw[proj](3,0,4)--(3,1.25,4);
% \draw[proj](2,0,4)--(2,1.25,4);
% \draw[proj](1,0,3)--(1,3,3);
% \draw[proj](1,0,2)--(1,3,2);
% \draw[ext](4,3,2)--(4,3,3)--(3,1.25,4)--(3,1.25,1)-- cycle;
% \draw[proj](4,0,2)--(4,3,2);
% \draw[proj](4,0,3)--(4,3,3);

% \draw[lin](2,0,1)--(2,0,4);
% \draw[lin](3,0,1)--(3,0,4);
% \draw[lin](1,0,3)--(4,0,3);
% \draw[lin](1,0,2)--(4,0,2);

% \draw[ext](1,3,3)--(4,3,3)--(4,3,2)--(1,3,2)-- cycle;
% \draw[ext](3,1.25,4)--(2,1.25,4)--(1,3,3)--(4,3,3)-- cycle;

% \draw (2.5,2,2.5) node{\large\color{lightblue}$Q$};
% \end{tikzpicture}
% \caption{Extended formulation of the regular octagon}
% \end{figure}

Let $G=(V,E)$ be a (finite, simple) graph with $n\coloneqq|V|$ and $m\coloneqq|E|$.  The \emph{stable set polytope} of $G$, denoted $\stab(G)$, is the convex hull of the characteristic vectors of stable sets of $G$. As previously mentioned, $\stab(G)$ can have very high extension complexity. In \cite{FMPTW15}, it is proved that if $G$ is obtained from a complete graph by subdividing each edge twice, then $\xc(\stab(G))$ is at least $2^{\Omega(\sqrt n)}$. Very recently, G{\"o}{\"o}s, Jain, and Watson \cite{GJW16} improved this to $2^{\Omega(n / \log n)}$, via a different class of graphs. For perfect graphs, Yannakakis \cite{Yannakakis91} proved an upper bound of $n^{O(\log n)}$, and it is an open problem whether Yannakakis' upper bound can be improved to a polynomial bound.

In this paper we restrict our attention to bipartite graphs.  Let $G=(V,E)$ be a bipartite graph with $n$ vertices, $m$ edges and no isolated vertices. By total unimodularity,
\[
\stab(G) = \{x \in \R^V \mid x_u \geqslant 0  \text{ for all } u \in V,\ x_u + x_v \leqslant 1  \text{ for all } uv \in E\}\,,
\]
and so $n \leqslant \xc (\stab(G)) \leqslant n+m$.  In this case $\xc (\stab(G))$ lies in a very narrow range, and it is a good test of current methods to see if we can improve these bounds.  

The situation is analogous to what happens with the spanning tree polytope of (arbitrary) graphs, 
where as previously mentioned, we also know that $\xc(\stp(G))$ lies in a very narrow range.
Indeed, a notorious problem of Goemans (see \cite{KT17}) is to improve the known bounds for $\xc(\stp(G))$, but this is still wide open.  

However, for the stable set polytopes of bipartite graphs, we are able to give an improvement.  Our main results are the following.

\begin{theorem}\label{thm:upperbound}
For all bipartite graphs $G$ with $n$ vertices, the extension complexity of
$\stab (G)$ is $O(n^2 / \log n)$.  
\end{theorem}

Note that Theorem \ref{thm:upperbound} is an improvement over the obvious upper bound when $G$ has quadratically many edges.

\begin{theorem}\label{thm:lowerbound}
  There exists an infinite class $\mathcal{C}$ of bipartite graphs such that every $n$-vertex graph in $\mathcal{C}$ has extension complexity $\Omega(n \log n)$. 
\end{theorem}

These are the first known examples of stable set polytopes of bipartite graphs where the extension complexity is more than linear in the number of vertices. For instance, $\xc (\stab (K_{n/2,n/2})) = \Theta(n)$. To the best of our knowledge, even for general perfect graphs $G$, the previous best lower bound for $\xc (\stab (G))$ was the trivial bound $|V(G)|$.

% We also provide examples where the obvious upper and lower bound are both essentially tight.

% {
% \renewcommand{\thetheorem}{\ref{thm:tightexample}}
% \begin{theorem}
%   There exists an infinite class $\mathcal{C}$ of bipartite graphs such that every $n$-vertex graph in $\mathcal{C}$ satisfies
% \[ 
% \frac{3n}{2} \leqslant \xc (\stab (G) ) \leqslant \frac{5n}{2}.
% \]
% \end{theorem}
% \addtocounter{theorem}{-1}
% }

% \begin{theorem} \label{thm:tightexample}
% There exists an infinite class $\mathcal{C}$ of bipartite graphs such that every $n$-vertex graph in $\mathcal{C}$ satisfies
% \[ 
% \frac{3n}{2} \leqslant \xc (\stab (G) ) \leqslant \frac{5n}{2}.
% \]
% \end{theorem}

\textbf{Paper Organization.} In Section \ref{sec:preliminaries} we define rectangle covers and fooling sets and we give examples of $3$-regular graphs with tight fooling sets. % In Section \ref{sec:tight} we prove Theorem \ref{thm:tightexample}. 
We prove Theorem \ref{thm:upperbound} in Section \ref{sec:upper bound} and Theorem \ref{thm:lowerbound} in Section \ref{sec:lower bound}.   In Section \ref{sec:limitations} we show that it is impossible to prove a better lower bound with the approach in Section \ref{sec:lower bound}.  Thus, to further improve the lower bound, different methods (or different graphs) are required.

\section{Rectangle Covers and Fooling Sets} \label{sec:preliminaries}

Consider a polytope $P \coloneqq \conv(X) = \{x \in \R^d \mid Ax \geqslant b\}$, where $X \coloneqq \{x^{(1)},\dots,x^{(n)}\} \subseteq \R^d$, $A \in \R^{m \times d}$ and $b \in \R^m$.
The \emph{slack matrix} of $P$ (with respect to the chosen inner and outer descriptions of the polytope) is the matrix $S \in \R^{m \times n}_{\geqslant 0}$ having rows indexed by the inequalities $A_1 x \geqslant b_1$, \ldots, $A_m x \geqslant b_m$ and columns indexed by the points $x^{(1)}$, \ldots, $x^{(n)}$, defined as $S_{ij}\coloneqq A_i x^{(j)} - b_i \geqslant 0$.

Yannakakis~\cite{Yannakakis91} proved that the extension complexity of $P$ equals the
nonnegative rank of $S$. In this work, we only rely on a lower bound that follows directly 
from this fact.
For a matrix $M$, we define the \emph{support} of $M$ as $\supp(M) \coloneqq \{(i,j) \mid M_{ij} \neq 0\}$. A \emph{rectangle} is any set of the form $R = I \times J$, with $R \subseteq \supp(M)$. A size-$k$ \emph{rectangle cover} of $M$ is a collection $R_1, \dots, R_k$ of rectangles such that $\supp(M) = R_1 \cup \dots \cup R_k$. The \emph{rectangle covering bound} of $M$ is the minimum size of a rectangle cover of $M$, and is denoted $\rc(M)$.

\begin{theorem} [Yannakakis, \cite{Yannakakis91}]\label{thm:Yannakakis2} 
Let $P$ be a polytope with $\dim(P) \geqslant 1$ and let $S$ be any slack matrix of $P$. Then, $\xc(P) \geqslant \rc(S)$.
\end{theorem}

A \emph{fooling set} for $M$ is a set of entries $F \subseteq \supp(M)$ such that $M_{i\ell} \cdot M_{kj} = 0$ for all distinct $(i,j), (k,\ell) \in F$. The largest size of a fooling set of $M$ is denoted by $\fool(M)$. Clearly, $\rc(M) \geqslant \fool(M)$.

% Let $G = (V,E)$ be a graph. The \emph{stable set polytope} of $G$ is defined as:
% \[
% \stab(G) \coloneqq \conv \{\chi^S \in \R^{V} \mid S \text{ stable set in }G\}.
% \]

% \begin{theorem}[Chv\'atal, 1975]
% If $G = (V,E)$ is perfect, then the stable set polytope of $G$ can be described as follows
% \[
% \stab(G) = \{ x \in \R^{V} \mid x(K) \leqslant 1,\text{ for every }K \text{ maximal clique}, x_v \geqslant 0, v \in V\}.
% \]
% \end{theorem}

% If $G=(V,E)$ is a bipartite graph, the maximal cliques in $G$ are precisely the edges, therefore $\stab(G)$ has the following simpler expression:
% \[
% \stab(G) = \{x \in \R^V \mid x_u \geqslant 0  \text{ for all } u \in V,\ x_u + x_v \leqslant 1  \text{ for all } uv \in E\}
% \]

Let $G$ be a bipartite graph. The \emph{edge vs stable set matrix} of $G$, denoted $M(G)$, is the $0/1$ matrix with a row for each edge of $G$, a column for each stable set of $G$, and a 1 in position $(e,S)$ if and only if $e \cap S=\varnothing$ (as usual, we regard edges as pairs of vertices). We say that $G$ has a \emph{tight fooling set} if $M(G)$ has a fooling set of size $|E(G)|$.  Note that if $G$ has a tight fooling set, then the non-negative rank of $M(G)$ is exactly $|E(G)|$.  Also observe that the property of having a tight fooling set is closed under taking subgraphs.  

It is easy to check that even cycles have tight fooling sets.  We now give an infinite family of $3$-regular graphs that have tight fooling sets.  A graph is \emph{$C_4$-free} if it does not contain a cycle of length four.  

\begin{theorem} \label{thm:tightexample}
Let $G = (V,E)$ be a $3$-regular, $C_4$-free bipartite graph. Then $G$ has a tight fooling set. 
\end{theorem}

\begin{proof}
For $X \subseteq V$, we let $\N(X)$ denote the set of neighbours of $X$. Let $V = A \cup B$ be a bipartition of the vertex set, and let $\phi : E \to \{1,2,3\}$ be a proper edge coloring of $G$, which exists by $3$-regularity and K\"onig's edge-coloring theorem (see e.g. \cite[Theorem 20.1]{S04}). For each vertex $a \in A$, we name its neighbors $a_1, a_2, a_3 \in B$ so that $\phi(aa_i) = i$. For each $a \in A$, consider the following stable sets: 
\[
\renewcommand\arraystretch{1.1}
\begin{array}{l}
S_{aa_1} \coloneqq A \setminus \{a\}\\
S_{aa_2} \coloneqq \{a_1\} \cup \{a' \in A \mid a' \notin \N(a_1) \}\\
S_{aa_3} \coloneqq B \setminus \{a_3\}\,.
\end{array}
\]
This defines a stable set $S_e$ disjoint from $e$, for every edge $e \in E$. Since $\phi$ is proper, no two of these stable sets are equal.
We claim that $\{(e,S_e) \mid e \in E\}$ is a fooling set in the edge vs stable set matrix of $G$.

Let $e$ and $f$ be distinct edges. We want to show that $S_e$ intersects $f$ or $S_f$ intersects $e$. Consider the following three cases. Let $e = aa_i$, where $i = \phi(e)$.\medskip

\emph{Case 1}. If $\phi(e) = 1$, then $S_e = S_{aa_1}$ intersects $f$ unless $f = aa_i$ for some $i \in \{2,3\}$. In both cases we have $a_1 \in S_f \cap e$.\medskip

\emph{Case 2}. If $\phi(e) = 3$, then $S_e = S_{aa_3}$ intersects $f$ unless $f = a'a_3$ for some $a' \in A$. Either $\phi(f) = 1$ and $S_f$ intersects $e$ (as in Case 1), or $\phi(f) = 2$. In the last case, since $G$ is $C_4$-free, we have $a \notin \N(a'_1)$. It follows that $S_f = S_{a'a_3} = S_{a'a'_2}$ intersects $e$.\medskip

\emph{Case 3}. If $\phi(e) = 2$, then we may also assume $\phi(f) = 2$ since otherwise by exchanging the roles of $e$ and $f$ we are back to one of the previous cases. Let $a'$ denote the endpoint of $f$ in $A$, so that $f = a'a'_2$. Because $\phi$ is proper, $a' \neq a$ and $a_1' \neq a_1$. Since $G$ is $C_4$-free, we have $a \notin \N(a'_1)$ or $a' \notin \N(a_1$). Hence, $a \in S_f \cap e$ or $a' \in S_e \cap f$.
\end{proof}

Note that there are infinitely many $3$-regular, $C_4$-free bipartite graphs.  For example, we can take a hexagonal grid on a torus.

% \begin{figure}[ht!]\centering
% \begin{tikzpicture}[scale=.9]

% \tikzstyle{every node}=[draw=black,circle,inner sep=2pt,minimum width=1.3pt]
% \tikzset{Anode/.style={fill=white},Bnode/.style={fill=black}}

% \def\wallheight{3}
% \def\wallwidth{3}

% \foreach \i in {0,...,\wallheight}{
% 	\foreach \j in {0,...,\wallwidth}{

% 		\ifthenelse{\i = 1 \OR \i= 3}{
% 			\def\s{1}
% 			\draw (2*\j+\s+2,\i)--(2*\j+\s,\i)--(2*\j+\s,\i+1)--(2*\j+\s+2,\i+1);
% 			\node [Anode] at (2*\j+\s,\i) {};
% 			\node [Anode] at (2*\j+\s+1,\i+1) {};
% 			\node [Bnode] at (2*\j+\s,\i+1) {};
% 			\node [Bnode] at (2*\j+\s+1,\i) {};
% 			\ifthenelse{\j = \wallwidth}{
% 				\draw (2*\j+\s+2,\i)--(2*\j+\s+2,\i+1);
% 				\node [Anode] at (2*\j+\s+2,\i) {};
% 				\node [Bnode] at (2*\j+\s+2,\i+1) {};
% 			}{}
% 		}{
% 			\ifthenelse{\j > 0}{
% 			\draw (2*\j+2,\i)--(2*\j,\i)--(2*\j,\i+1)--(2*\j+2,\i+1);
% 			\node [Anode] at (2*\j,\i) {};
% 			\node [Anode] at (2*\j+1,\i+1) {};
% 			\node [Bnode] at (2*\j,\i+1) {};
% 			\node [Bnode] at (2*\j+1,\i) {};	
% 			\ifthenelse{\j = \wallwidth}{
% 				\draw (2*\j+2,\i)--(2*\j+2,\i+1);
% 				\node [Anode] at (2*\j+2,\i) {};
% 				\node [Bnode] at (2*\j+2,\i+1) {};
% 				}{}
% 			}{
% 				\draw (1,0) -- (2,0);
% 				\draw (8,0) -- (9,0);
% 				\node [Anode] at (2,0) {};
% 				\node [Anode] at (8,0) {};
% 				\node [Bnode] at (1,0) {};
% 				\node [Bnode] at (9,0) {};
% 			}
% 		}
% 	}
% }
% \end{tikzpicture}
% \caption{A $2$-colouring of a toroidal grid.  The top and bottom vertices are identified as well as the left and right vertices.}\label{fig:grid}
% \end{figure}

\section{An Improved Upper Bound} \label{sec:upper bound}

In this section we prove Theorem \ref{thm:upperbound}. We use the following result of Martin \cite{Martin91}.

\begin{lem} \label{lem:martinduality}
If $Q$ is a nonempty polyhedron, $\gamma \in \R$, and
\[
P = \{x \mid \scalProd{x}{y} \leqslant \gamma \text{ for every } y \in Q\},
\]
then $\xc(P) \leqslant \xc(Q) + 1$.
\end{lem}

The \emph{edge polytope} $\edge(G)$ of a graph $G$ is the convex hull of the incidence vectors in $\R^{V(G)}$ of all edges of $G$.  The second ingredient we need is the following bound on the extension complexity of the edge polytope of all $n$-vertex graphs due to Fiorini, Kaibel, Pashkovich, and Theis \cite[Lemma 3.4]{FKPT13}.  This bound follows from a nice result of Tuza \cite{Tuza84}, which states that every $n$-vertex graph can be covered with a set of bicliques of total weight $O(n^2 / \log n)$, where the weight of a biclique is its number of vertices.  

\begin{lem} \label{lem:edgeupper bound}
For every graph $G$ with $n$ vertices, $\xc(\edge(G)) = O(n^2 / \log n)$. %the extension complexity of $\edge(G)$ is $O(n^2 / \log n)$.  
\end{lem}

% We are now in position to prove Theorem \ref{thm:upperbound}, which we restate for the reader's convenience.

% \begin{reptheorem}{thm:upperbound}
% For all bipartite graphs $G$ with $n$ vertices, the extension complexity of
% $\stab (G)$ is $O(n^2 / \log n)$.
% \end{reptheorem}

\begin{proof}[Proof of Theorem \ref{thm:upperbound}]
Let $G = (V,E)$. Since
\[
\stab(G) = \R_{\geqslant 0}^{V} \cap \{x \in \R^V  \mid \scalProd{x}{y} \leqslant 1 \text{ for every } y \in \edge(G)\},
\]
By Lemmas \ref{lem:martinduality} and \ref{lem:edgeupper bound}, the extension complexity of $\stab(G)$ is $O(n^2 / \log n)$. 
\end{proof}

\section{An Improved Lower Bound} \label{sec:lower bound}

In this section we prove Theorem \ref{thm:lowerbound}.  The examples we use to prove our lower bound are incidence graphs of finite projective planes.  We will not use any theorems from projective geometry, but the interested reader can refer to \cite{Coxeter94}. 

Let $q$ be a prime power, $\GF(q)$ be the field with $q$ elements, and $\PG(2,q)$ be the projective plane over $\GF(q)$.  The \emph{incidence graph} of $\PG(2,q)$, denoted $\mathcal I(q)$, 
is the bipartite graph with bipartition $(\points,\lines)$, where $\points$ is the set of points of $\PG(2,q)$, $\lines$ is the set of lines of $\PG(2,q)$, and $\pp \in \points$ is adjacent to $\LL \in \lines$ if and only if the point $\pp$ lies on the line $\LL$. For example, $\PG(2,2)$ and its incidence graph $\mathcal I(2)$ are depicted in Figure \ref{fig:fanoplane}.

\begin{figure}[ht!]\centering
\subfloat[]{\label{fig:2a}\parbox[c][3.8cm][c]{.45\textwidth}{\centering
\begin{tikzpicture}[scale=1.7]
\draw (0,0)--(2,0)--(1,1.732)--cycle;
\draw (2,0)--(.5,.866);
\draw (0,0)--(1.5,.866);
\draw (1,1.732)--(1,0);
\draw (1,.58) circle (.58cm);

{\tikzstyle{every node}=[fill=red,draw=black,circle,inner sep=2pt,minimum width=1pt]
\node[label={left:$a$}] at (0,0) {};
\node[label={left:$b$}] at (.5,.866) {};
\node[label={above:$c$}] at (1,1.732) {};
\node[label={right:$d$}] at (1.5,.866) {};
\node[label={right:$e$}] at (2,0) {};
\node[label={below:$f$}] at (1,0) {};
\node[label={[yshift=2mm]above right:$g$}] at (1,.58) {};}
{\tikzstyle{every node}=[text = blue]
\node[label={[label distance=2mm,xshift=2.5mm]above:\contour{white}{\color{blue}$1$}}] at (0,0) {};
\node[label={[yshift=-1.75mm,label distance=2mm]above left:\contour{white}{\color{blue}$5$}}] at (2,0) {};
\node[label={[label distance=2mm,xshift=-1.75mm]below right:\contour{white}{\color{blue}$2$}}] at (1,1.732) {};
\node[label={[label distance=2mm]below:\contour{white}{\color{blue}$4$}}] at (1,1.732) {};
\node[label={[label distance=2mm]left:\contour{white}{\color{blue}$3$}}] at (2,0) {};
\node[label={[yshift=-1.75mm,label distance=2mm]above right:\contour{white}{\color{blue}$6$}}] at (0,0) {};
\node[label={[yshift=6.5mm]above right:\contour{white}{\color{blue}$7$}}] at (1,.58) {};}
\end{tikzpicture}
}}\hspace*{2em}\subfloat[]{\label{fig:2b}\parbox[c][3.8cm][c]{.45\textwidth}{\centering
\begin{tikzpicture}[yscale=.6]
\draw (0,6)--(2,6);
\draw (0,6)--(2,4);
\draw (0,6)--(2,1);

\draw (0,5)--(2,6);
\draw (0,5)--(2,2);
\draw (0,5)--(2,0);

\draw (0,4)--(2,6);
\draw (0,4)--(2,5);
\draw (0,4)--(2,3);

\draw (0,3)--(2,5);
\draw (0,3)--(2,1);
\draw (0,3)--(2,0);

\draw (0,2)--(2,5);
\draw (0,2)--(2,4);
\draw (0,2)--(2,2);

\draw (0,1)--(2,3);
\draw (0,1)--(2,4);
\draw (0,1)--(2,0);

\draw (0,0)--(2,3);
\draw (0,0)--(2,2);
\draw (0,0)--(2,1);

{\tikzstyle{every node}=[fill=white,draw=black,circle,inner sep=2pt,minimum width=1pt]
\node[label={left:$g$}] at (0,0) {};
\node[label={left:$f$}] at (0,1) {};
\node[label={left:$e$}] at (0,2) {};
\node[label={left:$d$}] at (0,3) {};
\node[label={left:$c$}] at (0,4) {};
\node[label={left:$b$}] at (0,5) {};
\node[label={left:$a$}] at (0,6) {};

\node[label={right:$7$}] at (2,0) {};
\node[label={right:$6$}] at (2,1) {};
\node[label={right:$5$}] at (2,2) {};
\node[label={right:$4$}] at (2,3) {};
\node[label={right:$3$}] at (2,4) {};
\node[label={right:$2$}] at (2,5) {};
\node[label={right:$1$}] at (2,6) {};}

\end{tikzpicture}
}}
\caption{$\PG(2,2)$ and its incidence graph $\mathcal I (2)$.}\label{fig:fanoplane}
\end{figure}
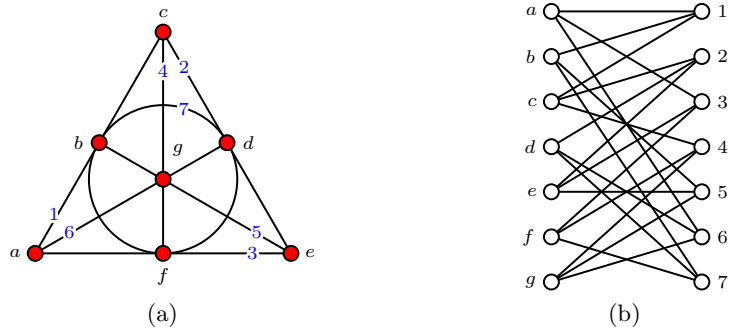

Before proving Theorem \ref{thm:lowerbound} we gather a few lemmas on binomial coefficients.  The first two are well-known, so we omit the easy proofs.    

\begin{lem} \label{lem:easybinomial}
For all integers $h$ and $c$ with $h \geqslant c \geqslant 0$ 
\[
\sum_{j=c}^h \binom{j}{c}=\binom{h+1}{c+1}.
\]
\end{lem}

\begin{lem} \label{lem:binomialidentity} For all positive integers $x,y$, and $h$,
\[
\sum_{j=0}^h \binom{x+j}{j}\binom{h+y-j}{h-j}=\binom{x + y + h + 1}{h}.
\]
\end{lem}

% \begin{proof}
% This identity is well-known, but for completeness, we include a proof.  Recall that $\binom{\ell+k-1}{k-1}$ is the number of ways to write $\ell$ as an ordered sum of exactly $k$ non-negative integers.  Thus, $\binom{\ell+k-1}{k-1}$ is the coefficient of $z^{\ell}$ in the power series $\frac{1}{(1-z)^k}$.  Therefore,
% \begin{eqnarray*}
% \binom{x + y + h + 1}{h} &=& \binom{x+y+h+1}{x+y+1} \\
% &=& [z^h] \frac{1}{(1-z)^{x+y+2}} \\
% &=&  [z^h] \left(\frac{1}{(1-z)^{x+1}}\right)\left(\frac{1}{(1-z)^{y+1}}\right) \\
% &=&  \sum_{j=0}^h \binom{x+j}{x}\binom{h+y-j}{y} \\
% &=& \sum_{j=0}^h \binom{x+j}{j}\binom{h+y-j}{h-j}. 
% \end{eqnarray*} 
% \end{proof}

\begin{lem}\label{lem:bound}
Let $q,c,t$ be positive integers with $c + t \leqslant q+1$. Then
\[
t\sum_{k=c}^{q+1-t}\frac{1}{k}\binom{q+1-t-c}{k-c} \binom{q}{k}^{-1} = \binom{t + c -1}{t}^{-1} \leqslant \frac{1}{c}.
\]
\end{lem}

\begin{proof} We have that
\begin{eqnarray*}
& & t \sum_{k=c}^{q+1-t}\frac{1}{k}\binom{q+1-t-c}{k-c} \binom{q}{k}^{-1} \\
&=& \frac{t (q+1-t-c)!}{q!}\sum_{k=c}^{q+1-t}\frac{(k-1)!(q-k)!}{(k-c)!(q+1-t-k)!} \\
&=& \frac{t (q+1-t-c)!}{q!}(c-1)!(t-1)!\sum_{k=c}^{q+1-t}\binom{k-1}{c-1}\binom{q-k}{t-1}.
\end{eqnarray*}
Moreover, 
\begin{eqnarray*}
\sum_{k=c}^{q+1-t}\binom{k-1}{c-1}\binom{q-k}{t-1} & = & \sum_{j=0}^{q+1-t-c}\binom{c -1 + j}{c - 1}\binom{q-c - j}{t-1} \\
\text{\footnotesize [$h=q+1-t-c$, $x = c-1$, $y=t-1$]} & = & \sum_{j=0}^h \binom{x+j}{j}\binom{h+y-j}{h-j} \\
\text{\footnotesize [by Lemma \ref{lem:binomialidentity}]} & = & \binom{x + y + h + 1}{h} \\
& = & \binom{q}{q+ 1 -t -c}.
\end{eqnarray*}
We conclude that
\begin{eqnarray*}
t \sum_{k=c}^{q+1-t}\frac{1}{k}\binom{q+1-t-c}{k-c} \binom{q}{k}^{-1} & = & \hspace{-2pt}\frac{t (q+1-t-c)!}{q!} \frac{q!(c-1)!(t-1)!}{(q+1-t-c)!(t+c-1)!}  \\
& = & \binom{t +c-1}{t}^{-1}.
 \end{eqnarray*}
The number of $t$-subsets of a set of size $t+c-1$ is at least $c$, since it includes all $t$-subsets containing a fixed set of size $t-1$. Hence, $\binom{ t +c-1}{t}^{-1} \leqslant \frac{1}{c}$.
\end{proof}

From the definition of $\PG(2,q)$ it follows that that $\mathcal I (q)$ is $(q+1)$-regular, $|V(\mathcal I (q))|=2(q^2+q+1)$, and $|E(\mathcal I (q))|=(q+1)(q^2+q+1)$.  Let $n=q^2+q+1$ and note that $\mathcal I (q)$ has $2n$ vertices. We let $\points$ and $\lines$ denote the set of points and lines of $\PG(2,q)$.  We also use the fact that $\mathcal I (q)$ is $C_4$-free. 

We denote the edge vs stable set incidence matrix of $\mathcal I (q)$ by $S_q$. 
Each 1-entry of $S_q$ is of the form $(e,S)$ where $e\in E$, $S\subseteq V$ is a stable set, and $e\cap S=\varnothing$. To prove Theorem \ref{thm:lowerbound} we will assign weights to the 1-entries of $S_q$ in such a way that the total weight is at least $\Omega (n \log n)$, while the weight of every rectangle is at most $1$.  The only entries that will receive non-zero weight are what we call \emph{special entries}, which we now define.

\begin{defn} \label{def:special}
A 1-entry of $S_q$ is \emph{special} if it has the form $(e,S(X))$ where 
\begin{itemize}
\item
$e=\pp \LL$ with $\pp\in \points, \LL \in \lines$,
\item 
$X\subseteq \N(\LL)\setminus \{\pp\}$, $X$ non-empty, 
\item 
$S(X)=X\cup (\lines\setminus \N(X))$.
\end{itemize} 
\end{defn}

We also need the following compact representation of maximal rectangles.

\begin{defn} \label{def:maxrectangles}
Let $R$ be a maximal rectangle. Then $R$ is determined by a pair $(\points_R, \lines_R)$ with $\points_R\subseteq \points$, $\lines_R \subseteq \lines$, where the rows of $R$ are all the edges between $\points_R$ and $\lines_R$ and the columns of $R$ are all the stable sets $S\subseteq V\setminus (\points_R\cup \lines_R)$.
\end{defn}

We are now ready to prove Theorem \ref{thm:lowerbound} in the following form.

% Moreover, $\mathcal I (q)$ is  $C_4$-free and, for every two points $p, p' \in P$, we have $|\N(p)\cap \N(p')|=1$, and similarly for any two lines $\ell,\ell' \in L$.

{
\renewcommand{\thetheorem}{\ref{thm:lowerbound}}
\begin{theorem}
Let $q$ be a prime power and $n=q^2+q+1$. Then there exists a constant $c>0$ such that 
\[
\xc (\stab (\mathcal I (q))) \geqslant c n \log n.
\]
\end{theorem}
\addtocounter{theorem}{-1}
}

\begin{proof}
Let $n=q^2+q+1$. Let $V=\points \cup \lines$ be the vertices of $\mathcal I (q)$, and $E$ be the edges of $\mathcal I (q)$.
To each special entry $(e,S(X))$ we assign the  weight 
\[
w(e,S(X))=\frac{1}{|X|\binom{q}{|X|} (q+1)}.
\]
All other entries of $S_q$ receive weight zero. 

\begin{clm} \label{firstfact}
$w(S_q)\coloneqq\sum_{(e,S)} w(e,S)\geqslant c n\log n$ for some constant $c$.
\end{clm}

\begin{subproof}
We have that
\begin{eqnarray*}
\sum_{(e,S)} w(e,S)=\sum_{(e,S(X)) \text{ special}} w(e,S(X)) &= &  \sum_{e\in E}\sum_{k=1}^q \binom{q}{k}\frac{1}{k\binom{q}{k} (q+1)} \\
& = & \frac{|E|}{q+1}\sum_{k=1}^q \frac{1}{k}= n \sum_{k=1}^q \frac{1}{k}> c n \log n.
\end{eqnarray*}
The claim follows.
\end{subproof}

Let $R=(\points_R, \lines_R)$ be an arbitrary maximal rectangle. 
We finish the proof by showing that $w(R)\coloneqq\sum_{(e,S)\in R} w(e,S)\leqslant 1$.  
Together with Claim \ref{firstfact} this clearly implies Theorem \ref{thm:lowerbound}. 
We will need the following obvious but useful Claim.

\begin{clm}\label{rem:cover_condition}
A special entry $(\pp \LL, S(X))$ is covered by $R=(\points_R, \lines_R)$ if and only if $X\cap \points_R=\varnothing$, $\lines_R\subseteq \N(X)$, $\pp\in \points_R$, and $\LL \in \lines_R$. %Indeed, it is immediate to see that these conditions are equivalent to $S(X)\cap (P_R\cup L_R)=\varnothing$.
\end{clm}

We consider two cases. First suppose that $\lines_R=\{\LL\}$ for some $\LL$. Then the only special entries covered by $R$ are of the form $(\pp \LL, S(X))$, with $X\subseteq \N(\LL)\setminus \points_R$. Let $\N(\LL)\cap \points_R=\{\pp_1,\dots,\pp_t\}$, where $1\leqslant t\leqslant q+1$. To compute $w(R)$ we have to sum over all edges $\pp_i \LL$ and over all subsets $X\subseteq \N(\LL)\setminus \{\pp_1,\dots,\pp_t\}$.  It follows that
\begin{eqnarray*}
w(R)&=&\sum_{i=1}^t \sum_{k=1}^{q+1-t} \binom{q+1-t}{k}\frac{1}{k\binom{q}{k} (q+1)} \\
&=& t\sum_{k=1}^{q+1-t} \frac{(q+1-t)!}{k! (q+1-t-k)!} \frac{k! (q-k)!}{k q! (q+1)} \\
&=& \frac{t (q+1-t)! (t-1)!}{(q+1)!}\sum_{k=1}^{q+1-t} \binom{q-k}{q+1-t-k}\frac{1}{k}\\
&=& \frac{1}{\binom{q+1}{t}} \sum_{k=1}^{q+1-t} \binom{q-k}{t-1}\frac{1}{k}\leqslant \frac{1}{\binom{q+1}{t}} \sum_{j=t-1}^{q-1} \binom{j}{t-1}=\frac{1}{\binom{q+1}{t}} \binom{q}{t}\leqslant 1,
\end{eqnarray*}
where the last equality follows from Lemma \ref{lem:easybinomial}. 

The remaining case is if $|\lines_R|\geqslant 2$.
For $\LL\in \lines_R$ such that $(\pp \LL,S(X))$ is covered by $R$ for some $\pp, X$, define 
\[
k_\LL=\min \{|X| \mid \text{ there exist } \pp, X : (\pp \LL,S(X)) \text{ is a special entry covered by } R\}.
\] 

%It is easy to convince oneself that $k_\ell$ does not depend on the choice of $p$.

\begin{clm}\label{rem:inclusion}
Let $(\pp \LL,S(X))$ be a special entry covered by $R$ such that $|X|=k_\LL$. Then for each $\pp',Y$ such that $R$ covers $(\pp' \LL,S(Y))$, we have $X\subseteq Y$.
\end{clm}

\begin{subproof}

 For each $\LL'\in \lines_R\setminus \{\LL\}$ (there is at least one since $|\lines_R|>1$),  we have $\LL'\in \N(X)$ by Claim \ref{rem:cover_condition}.  That is, there is $\xx=\xx(\LL')\in X$ adjacent to $\LL'$. Similarly, since $\LL'\in \N(Y)$, there is ${\mathsf y}={\mathsf y}(\LL')\in Y$ adjacent to $\LL'$. Now, if $\xx(\LL')\neq {\mathsf y}(\LL')$, then $\mathcal I (q)$ contains a $4$-cycle, which is a contradiction. Hence we must have $\xx(\LL')={\mathsf y}(\LL')$ for all $\LL'\in \lines_R\setminus \{\LL\}$. Now if there is an $\xx \in X$ such that $\xx\neq \xx(\LL')$ for every $\LL'\in \lines_R\setminus \{\LL\}$, then $(\pp \LL,S(X\setminus \{\xx\}))$ is still covered by $R$, contradicting the minimality of $X$. We conclude $X\subseteq Y$, as required.
 \end{subproof}

Now fix $\LL \in \lines_R$, and let 
\[
w(\LL)=\sum \{w(\pp \LL,S(X)) \mid (\pp \LL,S(X)) \text{ special}\}.
\]

\begin{clm} \label{clm:w(l)}
For every $\LL \in \lines_R$, 
\[
w(\LL)\leqslant \frac{1}{(q+1)k_\LL}.
\]
\end{clm}
\begin{subproof}
Let $\N(\LL)\cap \points_R=\{\pp_1,\dots,\pp_t\}$, where $1\leqslant t\leqslant q+1$. Let $X$ be such that $(\pp \LL,S(X))$ is a special entry covered by $R$ and $|X|=k_\LL$.  
By Claim \ref{rem:inclusion}, the only special entries appearing in the above sum are of the form $(\pp_i \LL,S(Y))$ where $i \in [t]$ and $X \subseteq Y \subseteq (\points \setminus \points_R) \cap N(\ell)$. Therefore 

%, when counting the sets $X$ such that $((p,\LL),S(X))$ is covered by $R$ for some $p$, we can fix a minimal $X$ of size $k_\LL$ and restrict the count the contribution of $((p,\ell),S(X))$ to all its supersets.  Let $w(\ell)$ be the total weight of entries of the form $((p,\ell),S(X))$ covered by $R$. We have:

\[
w(\LL)\leqslant t \sum_{k=k_\LL}^{q+1-t} \binom{q+1-t-k_\LL}{k-k_\LL} \frac{1}{k\binom{q}{k} (q+1)} \leqslant \frac{1}{(q+1)k_\LL},
\]
where the last inequality follows from Lemma \ref{lem:bound} with $c=k_\LL$.
\end{subproof}

%We need two more remarks to conclude the proof.
%\begin{remark}\label{onesidedmatching}
%If $R$ covers $((p,\ell),S(X))$ for some $X$, then $\N(p)\cap L_R=\{\ell\}$. 

%Indeed, if there is $\ell'\in \N(p)\cap L_R$, $\ell'\neq \ell$, then since by the previous remark $\ell' \in \N(X)$ so there is $x\in X$ adjacent to $\ell'$. But then $p,\ell,x, \ell'$ induce a 4-cycle, in contradiction with Observation \ref{obs:facts}.
%\end{remark}

\begin{clm}\label{rem:k_min}
For every $\LL \in \lines_R$,  $|\lines_R|\leqslant {(q+1)k_\LL}$.
\end{clm}
\begin{subproof}
Again, let $X$ be such that $(\pp \LL,S(X))$ is covered by $R$ and assume that $|X|=k_\LL$. By Claim \ref{rem:cover_condition}, we have $\lines_R\subseteq \N(X)$.

Hence $|\lines_R|\leqslant |\N(X)| \leqslant (q+1)|X|=(q+1)k_\LL$.
\end{subproof}

By Claim \ref{clm:w(l)} and Claim \ref{rem:k_min}, for every $\LL \in \lines_R$,
$w(\LL) \leqslant \frac{1}{|\lines_R|}.$  But clearly $w(R)=\sum_{\LL \in \lines_R} w(\LL)$, and so $w(R) \leqslant 1$, as required. This completes the entire proof. 
\end{proof}

\section{A small rectangle cover of the special entries}
\label{sec:limitations}
In this section we show that the submatrix of special entries considered in the previous section has a rectangle cover of size $O(n \log n)$.
Combined with Theorem \ref{thm:lowerbound}, this implies that a minimal set of rectangles that cover all the special entries always has size $\Theta (n \log n)$.
Thus, to improve our bound, we must consider a different set of entries of the slack matrix, or use a different set of graphs.  

This cover will be built from certain labeled trees which we now define.  Note that the \emph{length} of a path is its number of edges.  

\begin{defn}\label{def:tree}
For every integer $k \geqslant 1$, we build a tree $T(k)$ recursively:  
\begin{itemize}
\item The tree $T(1)$ consists of a root $r$ and a single leaf attached to it.%, see Figure \ref{fig:t1}.

\item For $k>1$, we construct $T(k)$ by first identifying one end of a path $P_1$ of length $k_1\coloneqq\left\lceil \frac{k}{2}\right\rceil$ to another end of a path $P_2$ of length $k_2\coloneqq\left\lfloor \frac{k}{2}\right\rfloor$ along a root vertex $r$. Let $\lambda_i$ be the 
end of $P_i$ that is not $r$.  We then attach a copy of $T(k_i)$ to $\lambda_{3-i}$, identifying $\lambda_{3-i}$ with the root of $T(k_i)$.
We call $P_1$ and $P_2$ the \emph{main} paths of $T(k)$. 
\end{itemize}
\end{defn}

% \begin{figure}[ht!]\centering
% \subfloat[$T(1)$]{\label{fig:t1}\parbox[c][3cm][c]{.3\textwidth}{\centering
% \begin{tikzpicture}[scale=.7]
% \draw(0,1)--(0,0);
% {\tikzstyle{every node}=[fill=white,draw,circle,inner sep=1.5pt,minimum width=0pt]
% \node at (0,1) {};
% \node[label={right:$r$}] at (0,0) {};}
% \end{tikzpicture}
% }}\subfloat[$T(2)$]{\label{fig:t2}\parbox[c][3cm][c]{.3\textwidth}{\centering
% \begin{tikzpicture}[scale=.7]
% \draw(-.5,2)--(-.5,1)--(0,0)--(.5,1)--(.5,2);
% {\tikzstyle{every node}=[fill=white,draw,circle,inner sep=1.5pt,minimum width=0pt]
% \node at (-.5,2) {};
% \node at (-.5,1) {};
% \node at (.5,1) {};
% \node at (.5,2) {};
% \node[label={right:$r$}] at (0,0) {};
% }
% \end{tikzpicture}
% }}\subfloat[$T(3)$]{\label{fig:t3}\parbox[c][3cm][c]{.3\textwidth}{\centering
% \begin{tikzpicture}[scale=.7]
% \draw(-1,3)--(-1,2);
% \draw(0,3)--(0,2)--(.5,1)--(1,2)--(1,3);
% \draw(-1,2)--(-.5,1)--(0,0);
% \draw(0,0)--(.5,1);
% {\tikzstyle{every node}=[fill=white,draw,circle,inner sep=1.5pt,minimum width=0pt]
% \node at (-1,3) {};
% \node at (-1,2) {};
% \node at (1,3) {};
% \node at (1,2) {};
% \node at (-.5,1) {};
% \node at (.5,1) {};
% \node at (0,3) {};
% \node at (0,2) {};

% \node[label={right:$r$}] at (0,0) {};
% }
% \end{tikzpicture}
% }}
% \caption{}\label{fig:ex_trees}
% \end{figure}

The next Lemma follows easily by induction on $k$.

\begin{lem}\label{lem:covtree}
For all $k\geqslant 1$,
\begin{enumerate}
    \item $T(k)$ has $O(k\log k)$ vertices;
    \item $T(k)$ has $k$ leaves;
    \item \label{lemitem:covtree3} every path from the root $r$ to a leaf has length $k$.
\end{enumerate}
\end{lem}

\begin{defn}\label{def:labeling}
We recursively define a labeling $\varphi_k : V(T(k)) \setminus \{r\} \rightarrow [k]$ as follows:
\begin{itemize}
\item Let $v$ be the non-root vertex of $V(T(1))$ and set $\varphi_1(v) \coloneqq 1$.  

\item For $k > 1$, let $P_1$ and $P_2$ be the main paths of $T(k)$. We name the vertices of $P_1$ as $r, v_1, \dots , v_{\left\lceil \frac{k}{2}\right\rceil}$ and $P_2$ as $r , v_{\left\lceil \frac{k}{2}\right\rceil+1} , \dots , v_{k}$, where these vertices are listed according to their order along $P_1$ and $P_2$.  Set $k_1 \coloneqq \left\lceil \frac{k}{2}\right\rceil$ and $k_2 \coloneqq \left\lfloor \frac{k}{2}\right\rfloor $. Note that $V(T(k)) = \bigcup_{i = 1,2} (V(P_i) \cup V(B_i))$, where $B_i$ is a copy of the tree $T(k_{3-i})$. We define 

\[
\varphi_k(v) =
\begin{cases}
i,  &\text{if } v=v_i \\
\varphi_{k_2}(v) + k_1,  &\text{if } v \in V(B_1) \setminus V(P_1) \\
\varphi_{k_1}(v),  &\text{if } v \in V(B_2) \setminus V(P_2)
\end{cases}\,
\]
\end{itemize}
% $\varphi_k : V(T(k)) \setminus \{r\} \rightarrow [k]$ by $\varphi_k(v_i) \coloneqq i$ for all $i \in [k]$ and $\varphi_k(v) \coloneqq \varphi_{\left\lfloor \frac{k}{2}\right\rfloor}(v) + \left\lceil \frac{k}{2}\right\rceil$ for every $v \in V(B_1) \setminus V(P_1)$ and $\varphi_k(v) \coloneqq \varphi_{\left\lceil \frac{k}{2}\right\rceil}(v)$ for every $v \in V(B_2) \setminus V(P_2)$.
% Where not explicitly stated, the labeling map $\varphi_k$ over the set of non-root vertices of $T(k)$ is simply denoted with $\varphi$. 
\end{defn}

\begin{figure}[ht!]\centering
\subfloat[$T(3)$ and the labeling $\varphi_3$]{\label{fig:t3}\parbox[c][5cm][c]{.46\textwidth}{\centering
\begin{tikzpicture}[scale=.7]
\draw(-1,3)--(-1,2);
\draw(0,3)--(0,2)--(.5,1)--(1,2)--(1,3);
\draw[blue](-1,2)--(-1,1)--(0,0);
\draw[red](0,0)--(.5,1);

%\node [xshift=-3mm,left,blue] at (-.5,1) {$P_1$};
\node [xshift=-3mm,left,blue] at (-1,1) {$P_1$};
\node [left] at (-1,2.5) {$B_1$};
\node [xshift=3mm,right,red] at (.25,.5) {$P_2$};
\node at (.5,2.5) {$B_2$};

{\tikzstyle{every node}=[fill=white,draw,circle,inner sep=1.5pt,minimum width=0pt]
\node[label={above:$3$}] at (-1,3) {};
\node[label={right:$2$},draw=blue] at (-1,2) {};
\node[label={above:$1$}] at (1,3) {};
\node[label={right:$2$}] at (1,2) {};
%\node[label={left:$1$},draw=blue] at (-.5,1) {};
\node[label={left:$1$},draw=blue] at (-1,1) {};
\node[label={right:$3$},draw=red] at (.5,1) {};
\node[label={above:$2$}] at (0,3) {};
\node[label={right:$1$}] at (0,2) {};

\node[label={right:$r$}] at (0,0) {};
}
\end{tikzpicture}
}}\hspace{-1mm}\subfloat[$T(8)$ and the labeling $\varphi_8$]{\label{fig:tq+1}\parbox[c][5cm][c]{.46\textwidth}{\centering
\begin{tikzpicture}[scale=.55]
\tikzset{every label/.append style={font=\scriptsize}}
\def\height{6}

%\foreach \i in {0,...,8}{\draw[line width=.4pt,dotted] (-.5,\height-\i)-- (10,\height-\i);}
\foreach \i in {0,...,3}{
\draw(2.5*\i,6)--(2.5*\i,5) -- (2.5*\i+1,4)-- (2.5*\i+2,5) -- (2.5*\i+2,6);
\draw (2.5*\i+1,4)-- (2.5*\i+1,3);}

\foreach \i in {0,1}{
\draw (1+5*\i,3)-- (1+5*\i+1.25,2)-- (1+5*\i+2.5,3);
\draw (1+5*\i+1.25,2)--(1+5*\i+1.25,-1);}

\draw (2.25,-1)--  (4.75,-2)-- (7.25,-1);

{\tikzstyle{every node}=[fill=white,draw,circle,inner sep=1.5pt,minimum width=0pt]

\foreach \i in {0,...,3}{
\node[label={above:$\pgfmathparse{2*\i+1}\pgfmathprintnumber{\pgfmathresult}$}] at  (7.5-2.5*\i+2,6) {};
\node[label={above:$\pgfmathparse{2*\i+2}\pgfmathprintnumber{\pgfmathresult}$}] at  (7.5-2.5*\i,6) {};
\node[label={above left:$\pgfmathparse{2*\i+2}\pgfmathprintnumber{\pgfmathresult}$}] at  (7.5-2.5*\i+2,5) {};
\node[label={above right:$\pgfmathparse{2*\i+1}\pgfmathprintnumber{\pgfmathresult}$}] at  (7.5-2.5*\i,5) {};

\node[label={right:$\pgfmathparse{8-\i}\pgfmathprintnumber{\pgfmathresult}$}] at  (6+1.25,2-\i) {};
\node[label={left:$\pgfmathparse{4-\i}\pgfmathprintnumber{\pgfmathresult}$}] at  (1+1.25,2-\i) {};
}

\foreach \i in {0,2}{
\node[label={left:$\pgfmathparse{6-2*\i}\pgfmathprintnumber{\pgfmathresult}$}] at  (2.5*\i+1,4) {};
\node[label={left:$\pgfmathparse{5-2*\i}\pgfmathprintnumber{\pgfmathresult}$}] at  (2.5*\i+1,3) {};
}

\foreach \i in {1,3}{
\node[label={right:$\pgfmathparse{10-2*\i}\pgfmathprintnumber{\pgfmathresult}$}] at  (2.5*\i+1,4) {};
\node[label={right:$\pgfmathparse{9-2*\i}\pgfmathprintnumber{\pgfmathresult}$}] at  (2.5*\i+1,3) {};
}

\node[label={below:$r$}] at  (4.75,-2) {};
}
\end{tikzpicture}
}}
\caption{}
\end{figure}

%We will assign to each non-root vertex $v$ of $T(k)$ a label $c(v)\in [k]$. The labels will represent the centers of the rectangles of our covering. To define a labeling of $T(k)$ recursively, we specify how to label the main paths and the sets of labels of the two branches. We label the vertices of $P_1$ in order with $1,\dots, \left\lceil \frac{k}{2}\right\rceil$ and the vertices of $P_2$ with $\left\lceil \frac{k}{2}\right\rceil+1,\dots, k$. Then we label recursively (the non-root nodes of) $B_1$ with the set labels of $P_2$ and $B_2$ with the set of labels of $P_1$. 
For each vertex $v \in T(k)$ we let $P(v)$ be the path in $T(k)$ from $r$ to $v$.  

\begin{lem}\label{lem:covtreelabel}
Let $\varphi_k$, $B_1$, and $B_2$ be as in Definition \ref{def:labeling}.
\begin{enumerate}
\item \label{lemitem:covtreelabel1} If $L$ is the set of leaves of $T(k)$, then $\varphi_k(L \cap V(B_1) )= \{\left\lceil \frac{k}{2}\right\rceil +1, \dots, k\}$ and $\varphi_k (L \cap V(B_2))= \{1,\dots,\left\lceil \frac{k}{2}\right\rceil\}$.
    
\item \label{lemitem:covtreelabel2} For every leaf $\lambda$ of $T(k)$, $\varphi_k(V(P(\lambda)) \setminus \{r\})=[k]$.
    
% \item \label{lemitem:covtreelabel3} For every pair of leaves $\lambda_1,\lambda_2$ of $T(k)$, let the vertices of $P(\lambda_i)$ be $r=v^i_0,v^i_1,\dots, v^i_k=\lambda_i$ (ordered from the root to the leaf). Then for every $j = 0,\dots,k-1$,
% \[
% \{\varphi_k(v^1_h)\mid j+1\leqslant h \leqslant k\} = \{\varphi_k(v^2_h)\mid j+1\leqslant h \leqslant k\}\quad\text{if and only if}\quad v^1_j = v^2_j\,.
% \]
%Equivalently, any two paths starting from the same vertex and ending in two leaves have the same set of labels.

\item \label{lemitem:covtreelabel4} Each label $i\in [k]$ occurs at most $\lceil\log k\rceil +1$ times in the labeling of $T(k)$.
\end{enumerate}
\end{lem}

\begin{proof} We proceed by induction on $k$. Property \ref{lemitem:covtreelabel1} follows directly from the recursive definition of the labeling $\varphi_k$. %, we observe that the sets of labels of the two branches $B_1$, $B_2$ are disjoint, hence the leaves of $B_1$ must have different labels than leaves of $B_2$. Now we reiterate this observation on the two branches of $B_1$ (and of $B_2$), and then on their branches, etc., until we consider branches that are copies of $T(1)$, to conclude that all the leaves must have a different label.

For \ref{lemitem:covtreelabel2}, let $\lambda$ be a leaf and let the (ordered) vertices of $P(\lambda)$ be $r,p_1,\dots,p_k=\lambda$.  Suppose that $\lambda \in V(B_i)$. Then $P(\lambda) \coloneqq P_i \cup P'$, where $P_i$ is a main path of $T(k)$ and $P'$ is the path in $B_i$ going from the root of $B_i$ to $\lambda$. Property \ref{lemitem:covtreelabel2} now follows by induction and the definition of $\varphi_k$.

%We claim that no two of them have the same label, this will imply \ref{lemitem:covtreelabel2}. Indeed, $P$ is the union of main paths of branches: by definition of our labeling, along each of those main paths labels do not repeat, and the labels of a main path are not used in any of the successive main paths.

% The ``$\Rightarrow$'' direction of Property \ref{lemitem:covtreelabel3} follows by \ref{lemitem:covtreelabel2} and induction. We now prove the ``$\Leftarrow$'' direction.  If $\lambda_1$, $\lambda_2$ belong to the same $B_i$, we are done by induction. So we may assume that $\lambda_i \in B_i$.  Thus, for all $j=1, \dots, \left\lfloor \frac{k}{2}\right\rfloor$, $1 \in \{\varphi_k(v^1_h)\mid j+1\leqslant h \leqslant k\}$, but $1 \notin \{\varphi_k(v^2_h)\mid j+1\leqslant h \leqslant k\}$. For $j=\left\lfloor \frac{k}{2}\right\rfloor+1, \dots, k-1$,
% \[
% \{\varphi_k(v^1_h)\mid j+1\leqslant h \leqslant k\} \neq \{\varphi_k(v^2_h)\mid j+1\leqslant h \leqslant k\},
% \]
% since they are actually disjoint.  

For \ref{lemitem:covtreelabel4}, first suppose that the label $i$ is in $[k_1]$. Then $i$ appears exactly once in the labeling of the main path $P_1$ of $T(k)$, it does not figure in the labeling of the nodes $V(P_2)\cup (V(B_1) \setminus V(P_1))$, and, by the inductive step, it occurs $\lceil{\log\lceil\frac{k}{2}\rceil}\rceil+1= \lceil{\log k}\rceil$ times in $\varphi_k(B_2)$. The thesis follows. A similar argument settles the remaining case $i \in [k]\setminus [k_1]$.
 %The number of nodes of $T(k)$ labelled with $i$ are: $1$ in the main path $P_1$ of $T(k)$, $0$ among nodes $V(P_2)\cup (V(B_1) \setminus V(P_1))$, and by the inductive step, $\lceil{\log\lceil\frac{k}{2}\rceil}\rceil+1= \lceil{\log k}\rceil$ in $B_2$. The thesis follows. A similar argument settles $i \in [k]\setminus [k_1]$.
\end{proof}

Henceforth, we simplify notation and denote the labeling $\varphi_k$ of $T(k)$ as $\varphi$.  
We now recall some notation from the previous section.  Let $q$ be a prime power and $S_q$ be the edge vs stable set incidence matrix of $\mathcal I (q)$. 

A maximal rectangle $R=(\points_R, \lines_R)$ is \emph{centered} if $|\lines_R| \geqslant 2$ and there is a point $\cc\in \points \setminus \points_R$ such that $\cc$ is incident to all lines in $\lines_R$.  We call $\cc$ the \emph{center} of $R$.  Note that the center is unique and its existence implies that $|\lines_R| \leqslant q+1$.  

One way to create centered rectangles is as follows.  Let $\LL$ be a line, $\cc$ be a point on $\LL$, and $Y \subseteq \N (\LL)$ with $\cc \in Y$. We let $\boxed{\cc, \LL, Y}$ be the centered rectangle $R = (\points_R,\lines_R)$ where $\points_R=\N(\LL) \setminus Y$ and $\lines_R=\N(\cc)$. Note that a special entry of the form $(\pp \LL,S(X))$ is covered by the centered rectangle $\boxed{\cc, \LL, Y}$ if and only if $\pp\notin Y$ and $\cc\in X \subseteq Y$.

We now fix a line $\LL \in \PG (2,q)$ and let $\N(\LL)=\{\pp_1,\dots, \pp_{q+1}\}$.
We will use the labeling $\varphi$ of $T(q+1)$ to provide a collection of centered rectangles that cover all special entries of the form $(\pp \LL,S(X))$. Recall that for a vertex $v$ of $T(q+1)$, $P(v)$ denotes the path in $T(q+1)$ from $r$ to $v$. If $v$ is neither the root nor a leaf of $T(q+1)$, we define
\[
Y(v)\coloneqq \{\pp_{\varphi(u)} \mid \text{ $u$ is a non-root vertex of $P(v)$}\}. 
\]
%By Lemma \ref{lem:covtreelabel}% \ref{lemitem:covtreelabel3}, $Y(v)$ is well defined.
%Let $T_\ell$ be the tree obtained from $T(q+1)$ by deleting its leaves. 
% As in Lemma \ref{lem:covtreelabel}
% \ref{lemitem:covtreelabel1}, given a labeling $\varphi$ of $T(q+1)$, it turns out that $\varphi$ is a bijection between the family of all leaves of $T_\ell$ and the labels in $[q+1]$. 
% \fix{This is false unless $q+1$ is a power of $2$.}

\begin{lem} \label{lem:localcover}
Fix a line $\LL\in \PG (2,q)$ and let $\N(\LL)=\{\pp_1,\dots, \pp_{q+1}\}$.
Let $\mathcal{R}_\LL$ be the collection of all centered rectangles $\boxed{\pp_{\varphi(v)}, \LL, Y(v)}$ where $v$ ranges over all non-root, non-leaf vertices of $T(q+1)$. Then every special entry $(e,S)$ with $\ell$ incident to $e$ is covered by some rectangle $R\in \mathcal{R}_\LL$. 
\end{lem}

\begin{proof}
% By Lemma \ref{lem:covtreelabel} \ref{lemitem:covtreelabel4}, every label $i \in [q+1]$ appears at most $\lceil\log(q+1)\rceil$ times in restriction of $\varphi_{q+1}$ to $T_\ell$.

Let $(\pp_i \LL,S(X))$ be such a special entry and let $\lambda$ be the (unique) leaf of $T(q+1)$ such that $\varphi(\lambda) = i$.  Name the vertices of $P(\lambda)$ as $r, u_1, \dots,u_{q+1}=\lambda$ (ordered away from the root).

Define $j = \max \{i \mid \pp_{\varphi(u_i)} \in X\}$. Since $\pp_{\varphi(\lambda)}\notin X$, note $j < q+1$.  By Lemma \ref{lem:covtreelabel}, $X  \subseteq Y(u_j)$. Also, by construction, $\pp_{\varphi(u_j)} \in X$ and $\pp\notin Y(u_j)$. We conclude that the centered rectangle $\boxed{\pp_{\varphi(u_j)}, \LL, Y(u_j)}$ covers the special entry $(\pp_i \LL,S(X))$, as required.
\end{proof}

By Lemma \ref{lem:localcover}, for each line $\LL$, there is a set $\mathcal{R}_\LL$ of $O(q\log q)$ centered rectangles that cover all special entries of the form $(\pp \LL,S(X))$. By taking the union of all $\mathcal{R}_\ell$, we get a cover $\mathcal{R}$ of size $O(n q\log q)$ for all the special entries. To prove the main theorem of this section, we now reduce the size of $\mathcal{R}$ by a factor of $q$.

\begin{theorem} \label{thm:smallcover}
There is a set of $O(n\log n)$ centered rectangles that cover all the special entries.
\end{theorem}

\begin{proof}
If $R_1 \coloneqq \boxed{\cc, \LL_1, Y_1}, \dots ,R_k \coloneqq \boxed{\cc, \LL_k, Y_k}$ are centered rectangles with the same center $\cc$, we let $\sum_{i=1}^k R_i=R$ be the maximal rectangle with $\points_{R}= \bigcup_{i=1}^k \N (\LL_i) \setminus \bigcup_{i=1}^k Y_i$ and $\lines_R=\N(\cc)$.  Note that $\sum_{i=1}^k R_i$ is also a centered rectangle with center $\cc$.

\begin{clm} \label{clm:rectangleadding}
If $R_1 \coloneqq \boxed{\cc, \LL_1, Y_1}, \dots ,R_k \coloneqq \boxed{\cc, \LL_k, Y_k}$ are centered rectangles such that $\LL_1, \dots, \LL_k$ are all distinct, then $\sum_{i=1}^k R_i$ covers all special entries covered by $\bigcup_{i=1}^k R_i$. 
\end{clm}

 \begin{subproof}
 Let $(\pp \LL, S(X))$ be a special entry covered by some $\boxed{\cc, \LL_j, Y_j}$. Clearly $\cc\in X\subseteq Y_j \subseteq \bigcup_{i=1}^k Y_i$. By contradiction, suppose $\pp\in \bigcup_{i=1}^k Y_i$. Since $\pp\notin Y_j$, $\pp\in Y_{j'} \subseteq \N(\LL_{j'})$ for some $j' \neq j$. But then $\cc \LL_j \pp\LL_{j'} $ is a $4$-cycle in $\mathcal I (q)$, which is a contradiction. Hence the entry $(\pp \LL, S(X))$ is also covered by $\sum_{i=1}^k R_i$.\end{subproof}
 
We iteratively use Claim \ref{clm:rectangleadding} to reduce the number of rectangles in our covering $\mathcal{R}$. 
For each point $\cc$, name the $q+1$ lines through $\cc$ as $\LL,\LL_1,\dots, \LL_{q}$, so that among $\mathcal{R}_\LL, \mathcal{R}_{\LL_1}, \dots, \mathcal{R}_{\LL_q}$, the collection $\mathcal{R}_\LL$ has the most rectangles with center $\cc$. Note that, by Lemma \ref{lem:covtreelabel}, $\mathcal{R}_\LL$ contains $O(\log q)$ rectangles with center $\cc$.  

Fix $i \in [q]$ and for each rectangle $R \in \mathcal{R}_{\LL_i}$ with center $\cc$ choose a  rectangle $f_i(R)$ with center $\cc$ in $\mathcal{R}_\LL$ such that $f_i(R) \neq f_i(R')$ if $R \neq R'$. For each $R \in \mathcal{R}_{\LL}$ we let 
\[
f^{-1}(R)=\{R\} \cup \bigcup_{i=1}^q \{R' \in \mathcal{R}_{\LL_i} \mid f_i(R')=R \}.
\]

We then remove all rectangles with center $\cc$ that appear in $\mathcal{R}_{\LL}, \mathcal{R}_{\LL_1}, \dots , \mathcal{R}_{\LL_q}$ and replace them with all rectangles of the form  
$\sum_{R' \in f^{-1}(R)} R'$, where $R$ ranges over all rectangles in $\mathcal{R}_{\LL}$ with center $\cc$.
In doing so, we obtain at most $O(\log q)=O(\log n)$ rectangles with center $\cc$. Repeating for every $\cc\in \points$ gives us $O(n\log n)$ rectangles in total.
\end{proof}

% \section{Conclusion} \label{sec:conclusion}

% The property of having a fooling set of size $|E|$ in the edge vs stable set matrix is monotone under taking subgraphs. It would be interesting to see what are the (bipartite) minimal forbidden subgraphs for this property. One of them is $K_{2,3}$. To prove that it does not have a fooling set of size 6, one can notice that there is a rectangle cover of size $5$: for each vertex $v$, consider the rectangle $(\delta(v), \mathcal S_v)$ where $\mathcal S_v$ contains all the stable sets that do not intersect $v$ or its neighbors. Moreover, it is easy to construct a fooling set for the graph obtained from $K_{2,3}$ by removing one edge, hence $K_{2,3}$ is minimal.

% \fix{What do you think of this?}

\textbf{Acknowledgement.} We thank Monique Laurent and Ronald de Wolf for bringing the topic of this paper to our attention. 
We also acknowledge support from ERC grant \emph{FOREFRONT} (grant agreement no. 615640) funded by the European Research Council under the EU's 7th Framework Programme (FP7/2007-2013) and \emph{Ambizione} grant PZ00P2 154779 \emph{Tight formulations of 0-1 problems} funded by the Swiss National Science Foundation.  Finally, we also thank the five anonymous referees for their constructive comments. 

%Yuri Faenza's research was partially supported by an \emph{Ambizione} grant PZ00P2 154779 \emph{Tight formulations of 0-1 problems} funded by the Swiss National Science Foundation.

\bibliography{bipartite}{}
\bibliographystyle{plain}

\end{document}